\documentclass[runningheads]{llncs}
\usepackage[T1]{fontenc}
\usepackage{graphicx}

\usepackage[numbers]{natbib}
\usepackage{amsmath}
\usepackage{amssymb}
\usepackage{booktabs}
\usepackage{tikz}
\usepackage{xcolor}
\usepackage{complexity}
\usepackage{environ}
\usepackage{algorithm}
\usepackage{algpseudocode}
\usepackage{tabularx}
\usepackage{complexity}
\usepackage{enumitem}
\usepackage{thm-restate}
\usepackage{todonotes}
\usepackage{hyperref}
\usepackage{cleveref}

\definecolor{CoalitionColor}{HTML}{1D42A6}

\usetikzlibrary{calc}
\usetikzlibrary{decorations.pathreplacing}
\tikzstyle{agent}=[draw, shape=circle, fill=white, fill opacity=1, text opacity=1, minimum size=30pt]
\tikzstyle{smallagent}=[draw, shape=circle, minimum size=15pt]
\tikzstyle{aux}=[draw, shape=circle, minimum size=5pt]

\newcommand{\NW}{\mathcal{NW}}

\newcommand{\convexpath}[2]{
  [
  create hullcoords/.code={
    \global\edef\namelist{#1}
    \foreach [count=\counter] \nodename in \namelist {
      \global\edef\numberofnodes{\counter}
      \coordinate (hullcoord\counter) at (\nodename);
    }
    \coordinate (hullcoord0) at (hullcoord\numberofnodes);
    \pgfmathtruncatemacro\lastnumber{\numberofnodes+1}
    \coordinate (hullcoord\lastnumber) at (hullcoord1);
  },
  create hullcoords
  ]
  ($(hullcoord1)!#2!-90:(hullcoord0)$)
  \foreach [
  evaluate=\currentnode as \previousnode using \currentnode-1,
  evaluate=\currentnode as \nextnode using \currentnode+1
  ] \currentnode in {1,...,\numberofnodes} {
    let \p1 = ($(hullcoord\currentnode) - (hullcoord\previousnode)$),
    \n1 = {atan2(\y1,\x1) + 90},
    \p2 = ($(hullcoord\nextnode) - (hullcoord\currentnode)$),
    \n2 = {atan2(\y2,\x2) + 90},
    \n{delta} = {Mod(\n2-\n1,360) - 360}
    in 
    {arc [start angle=\n1, delta angle=\n{delta}, radius=#2]}
    -- ($(hullcoord\nextnode)!#2!-90:(hullcoord\currentnode)$) 
  }
}

\makeatletter
\newcommand{\problemtitle}[1]{\gdef\@problemtitle{#1}}\newcommand{\probleminput}[1]{\gdef\@probleminput{#1}}\newcommand{\problemsolution}[1]{\gdef\@problemsolution{#1}}\NewEnviron{problemdef}{
	\problemtitle{}\probleminput{}\problemsolution{}\BODY \par\addvspace{.5\baselineskip}
	\noindent
	\begin{tabularx}{\textwidth}{@{\hspace{\parindent}} l X c}
		\toprule
		\multicolumn{2}{@{\hspace{\parindent}}l}{\@problemtitle} \\\textbf{Input:} & \@probleminput \\\textbf{Answer:} & \@problemsolution \\\bottomrule
	\end{tabularx}
	\par\addvspace{.5\baselineskip}
}
\makeatother

\usepackage{color}

\begin{document}
\title{Nash~Welfare in Additively~Separable~Hedonic~Games}
\author{Marta Pagano\inst{1,2} \and Alexander Schlenga\inst{2}}
\authorrunning{M. Pagano and A. Schlenga}
\institute{Sapienza University of Rome, Italy \and Technical University of Munich, Germany\\
\email{alexander.schlenga@tum.de}
}

\maketitle

\begin{abstract}
Additively separable hedonic games (ASHGs) are a prominent model of coalition formation where agents' preferences are derived from their individual valuations of peers. While social welfare maximization in ASHGs has traditionally focused mostly on utilitarian welfare, Nash welfare---a well-established metric in economics which balances fairness with efficiency and offers scale invariance---has been entirely overlooked. In this paper, we initiate the study of Nash welfare in ASHGs. We point out desirable properties fulfilled by partitions with high Nash welfare. This includes
guaranteed contractual Nash stability in symmetric games,
even for any approximation of Nash welfare.
This is particularly appealing since, as for other welfare notions, Nash welfare turns out to be \NP-hard to maximize, even for the ASHG subclass of symmetric aversion to enemies games (AEGs). A main focus of our study is on approximation algorithms for the Nash welfare objective. We present packing-based algorithms with approximation ratios
for well-established subclasses of ASHGs: $n-1$ for AEGs and $2n$ for appreciation of friends games. This is complemented by a strict inapproximability result showing it is \NP-hard to approximate Nash welfare within a factor of $1.0000759$ in general ASHGs. Further, we investigate the restricted settings with an upper bound on the coalition size or number of coalitions, and draw the boundary between the cases admitting efficient algorithms and those yielding \NP-hardness: bounding the allowed size or number of coalitions by $2$ admits polynomial-time solvability, whereas bounds of $3$ or more yield \NP-hardness or unbounded inapproximability.

\end{abstract}

\section{Introduction}

The formation of coalitions is a well-studied research topic and has been receiving attention for a long time.
It has already been present in the pioneering work on game theory by \citet{vNM53a}.
In the past decades, a plethora of models and solution concepts have been studied to tailor coalition formation models for real word applications \citep{Ray07a,AzSa15a,BER24a}.

Most models focus on settings in which everyone can only be a member of exactly one coalition.
One of the most prominent such models are hedonic games \citep{DrGr80a}.
They give agents the possibility to express unrestricted preferences over the members of their coalition, while making the natural assumption that the preferences are independent of the arrangements of other agents outside their coalition.

Different approaches to what makes a good solution concept for hedonic games have been proposed and investigated, depending on the desired tradeoff between efficiency, fairness, strategyproofness, and algorithmic tractability.
Many of them are adapted from related problems in (algorithmic) economics.
It is commonly agreed that partitions into coalitions should satisfy desirable properties, such as various notions of stability or welfare optimality. Unfortunately, when we are interested in whether partitions with such desirable properties exist, the ability to express arbitrary preferences can cause severe computational obstacles, both with regards to (space-)efficiently stating preferences as well as computing solution concepts.

In order to create a more practical model, additively separable hedonic games (ASHGs) were proposed as a way of concisely inferring an agent's preferences over coalitions from her valuations for other single agents \citep{BoJa02a}. This model has since been subject to a large body of research, in economics as well as distributed artificial intelligence and other areas of computer science \citep{BER24a}.
A prominent desideratum is to maximize social welfare \citep{BoJa02a,ABS11c,FKV22a,BCS25a,FMM+21a,BuRo25b,CoAg25a}.
Such global goals are especially suitable when they can be implemented by a central coordinating instance, but they can even lead to stability guarantees in a distributed view based on incentives of single agents \citep{BoJa02a}.

Most of this work focuses on utilitarian welfare, with egalitarian \citep{ABS11a,CoAg25a} and elitist \citep{ABS11a} welfare receiving some attention.
Somewhat surprisingly, Nash welfare has not been considered so far in ASHGs. Even beyond ASHGs, in the realm of hedonic games, only an isolated result concerning computational hardness in the related model of fractional hedonic games (FHGs) is known \citep{AGG+15b}.
Our paper lays the foundation for closing this gap in the literature.
The motivation behind this is threefold.
\begin{enumerate}
    \item Nash welfare is a well-established social welfare function in economics \citep{Nash50b,Hars59a,KaNa79a}, often praised to achieve just the right balance between efficiency\footnote{Here, by ``efficiency'' we refer to the sum of agents' utilities, i.e., utilitarian welfare.} and fairness (egalitarian and\slash{}or envy-based) \citep{CKM+19a}, and in some models even exhibiting good properties beyond that, like strategyproofness \citep{BBP+19a}.
    So far, fairness and, even more so, the balance between fairness and efficiency have received little attention in ASHGs.
The Nash welfare seeks to avoid singletons (as otherwise it becomes $0$) which feels natural in a coalition formation setting.
    \item Nash welfare is invariant under scaling agents' utilities.
    This makes it a very appealing solution concept for ASHGs without bounded valuation functions.
    Agents cannot compete for priority in consideration by ``naming the bigger number'' when asked for their utilities, a weakness which utilitarian welfare does suffer from.
Moreover, for agents, being in a singleton coalition serves as a natural baseline for utility---agents can influence their attributed utility in proper coalitions by adapting their valuations towards the others but they cannot change the fact that by being alone they receive $0$ utility.
    Here, the all-singleton partition corresponds to what \citet{KaNa79a} call the ``origin''---a state of society from which everyone should be able to improve.
    As a result, for Nash welfare, we can consider agents' utility functions to be normalized.
    \item Nash welfare has received quite a lot of attention for being a very good solution concept in models of fair division \citep{DaSc15a,CKM+19a,ABF+21a,GaMu23a,Suks23a} and matching \citep{JaVa24a,GSVY25a} which can be seen as restricted subclasses of ASHGs.
    Hence, studying it in this overarching general model connects related fields and broadens our understanding of Nash welfare not only for ASHGs.
\end{enumerate}

\subsection{Related Work}

Hedonic games have been conceptualized and first studied by \citet{DrGr80a}.
Later, they were popularized by \citet{BoJa02a} and \citet{BKS01a}, who also introduced ASHGs.
They studied the core and other stability notions.
In particular, \citet{BoJa02a} showed that in symmetric ASHGs, a partition with optimal utilitarian welfare is always Nash stable (a different concept than Nash welfare).
\Citet{AMT+95a} have independently proposed the essentials of ASHGs from a political science perspective.
For an overview of the literature on hedonic games, see the book chapters by \citet{AzSa15a} or \citet{BER24a}.

Subsequently, \citet{CeHa02a} and \citet{Ball04a} initiated the study of computational complexity in hedonic games, showing the first boundaries between efficient algorithms and \NP-hardness.
\Citet{SuDi10a} were the first to investigate computational aspects of ASHGs.
This line of research was continued by \citet{ABS11c} who clarified the computational aspects for a wide range of solution concepts in ASHGs.
Among other things, they showed that computing partitions with optimal utilitarian or egalitarian welfare is \NP-hard.
Only maximizing elitist welfare turned out to be achievable in polynomial time.
\Citet{PeEl15a} and \citet{ZeBu26a} introduced general frameworks for proving computational hardness results for stability notions in a wide range of different hedonic games models, ASHGs among them.

Other succinctly representable classes of hedonic games were proposed, among them aversion to enemies games (AEGs) and appreciation of friends games (AFGs) by \citet{DBHS06a}, which are subclasses of ASHGs and natural models we focus on in our paper.
\Citet{ABB+17a} and \citet{Olse12a} defined FHGs and modified fractional hedonic games (MFHGs), respectively, which are akin to ASHGs but utilities are derived by taking an average (instead of the sum) over single-agent valuations.
\Citet{AGG+15b} studied complexity and approximation of different welfare notions in FHGs, Nash welfare\footnote{They use the product of the utilities but this is equivalent up to a monotone transformation.} among them.
They proved all considered welfare functions to be \NP-hard to maximize and provided approximation algorithms for utilitarian and egalitarian welfare.
\Citet{EFF20a} and \citet{Bull19a} studied Pareto-optimality, a natural ordinal relaxation of welfare maximality, and its computational aspects in ASHGS, FHGs, and MFHGs.
Also here, many problems turned out the be \NP-hard in general, and restricted classes were to be considered for positive results.

\Citet{FKV22a} investigated strategyproofness in algorithms maximizing utilitarian welfare in AEGs and AFGs.
Moreover, they showed \NP-hardness of approximating utilitarian welfare for AEGs within $\mathcal{O}(n^{1-\varepsilon})$ for any constant $\varepsilon>0$.
This has been picked up by \citet{BCS25a}, who studied in detail the approximability of utilitarian welfare in ASHGs.
Both works showed that randomization helps, and partitions with good utilitarian welfare approximations in expectation are possible.
Recently, \citet{FMM+21a}, \citet{BuRo25b}, \citet{CoAg25a}, and \citet{BRS26a} started a line of research on online algorithms for welfare maximization in hedonic games and studied the achievable competitive ratios for ASHGs and FHGs.

Instead of expressing preference for all possible partition, a natural assumption is to restrict the size of coalitions.
Regarding ASHGs, \citet{FGM25a} provide a more nuanced picture of the parameterized complexity of welfare maximality, and \citet{LHSA24a} and \citet{BDEG26a} study stability.
It is also common to consider both a bound on coalition sizes as well as the number of coalitions \citep[see, e.g.,][]{FMM+21a}.

We proceed by discussing the origins of Nash welfare.
It was first introduced in the seminal paper by \citet{Nash50b}, where he carved out the product of two agents' utilities as the unique solution concept satisfying a set of desirable axioms in a bargaining model.
Then, \citet{Hars59a} generalized this concept to situations with $n$ agents.
Later, \citet{KaNa79a} picked it up and defined the Nash (social) welfare function by the product of agents' utilities.\footnote{They use the sum of logarithms of the utilities but this is equivalent up to a monotone transformation.
}

Nash welfare has received a lot of attention in fair division. We conclude this section by listing the corresponding works which are most relevant to ours. Computational aspects have been studied by \citet{DaSc15a}, and further by \citet{CoGz18a} who proved there exists a constant factor approximation, complemented by \citet{Lee17a} who showed \APX-hardness. \Citet{GaMu23a} deepened this study of approximability in different aspects. Inspired by this line of research, \citet{JaVa24a} and \citet{GSVY25a} started the study of Nash welfare computation and approximation in matching with two-sided cardinal utilities (which can be seen as one step towards the more general model of ASHGs). Axiomatic properties of the maximum Nash welfare solution in fair division have been studied by \citet{CKM+19a} who showed that is fair in different aspects, in particular it always constitutes an EF1 allocation, which was later used by \citet{Suks23a} for a characterization, while \citet{ABF+21a} showed that it even satisfies the stronger notion of EFX if there are at most two distinct valuations.

\subsection{Our Contribution}

We are the first to define and study Nash welfare in the setting of ASHGs.
Most of our work concerns its computational aspects.
First, we describe at the end of \Cref{sec:prelims} why we define the Nash welfare as the geometric mean (instead of the product) of agent utilities.
Then, we begin \Cref{sec:aeg} with setting a baseline by showing that it is \NP-hard to compute a partition of optimal Nash welfare, even for the restricted subclass of symmetric AEGs.
We proceed to provide a simple packing-based algorithm for the same setting, achieving an approximation ratio of $n-1$.
Moreover, we give insights into the structure of partitions of high Nash welfare and characterize when the optimal Nash welfare is $0$.
In \Cref{sec:afg}, we move to the related setting of AFGs and answer similar questions.
Here, we present a more sophisticated algorithm for achieving an approximation ratio of $2n$.
We complement our approximation guarantees by showing in \Cref{sec:hardness} that there exists a constant by which it is \NP-hard to approximate Nash welfare in ASHGs.
For results of a different flavor, we investigate in \Cref{sec:bound_number,sec:bound_size} what happens if we enforce an upper bound on the number or size of coalitions.
Here, we observe the typical dichotomy between the sizes of $2$ and $3$ in computational complexity: For bounds of $2$, we describe optimal efficient algorithms, whereas for bounds of at least $3$, we show hardness (for the bound on the number of coalitions even hardness of any approximation).
Lastly, in \Cref{sec:stability}, we investigate the stability guarantees implied by achieving optimal or approximate Nash welfare but also show aspects in which optimal Nash welfare is weaker than utilitarian welfare in that regard.
All proofs are deferred to the appendix.

\section{Preliminaries}\label{sec:prelims}

Let $N$ be a finite set of \emph{agents}.
A nonempty subset $C\subseteq N$ is called a \emph{coalition}.
The set of coalitions containing agent~$i\in N$ is denoted by $\mathcal N_i:=\{C\subseteq N\mid i\in C\}$.
A set $\pi$ of disjoint coalitions containing all members of $N$ is a \emph{partition} of $N$.
For agent $i\in N$ and partition $\pi$, let $\pi(i)$ denote the unique coalition in~$\pi$ that~$i$ belongs to.

A (cardinal) \emph{hedonic game} is a pair $G = (N,u)$ where $N$ is the set of agents and $u = (u_i)_{i\in N}$ is a tuple of \emph{utility functions} $u_i\colon \mathcal N_i \to \mathbb Q$.
Since, in our work, the considered game is always fixed, we omit $G$ as a parameter in our notations.
Agents seek to maximize utility and prefer partitions in which their coalition gives them a higher utility.
Hence, we define the utility of a partition~$\pi$ for agent~$i$ as $u_i(\pi) := u_i(\pi(i))$.
We denote by $n:=|N|$ the number of agents.

Following \citet{BoJa02a}, an \emph{ASHG} is a hedonic game $(N,u)$, where for each agent $i\in N$ there exists a \emph{valuation} function $v_i\colon N\setminus\{i\}\to \mathbb Q$ such that for all $C\in \mathcal N_i$ it holds that $u_i(C) = \sum_{j\in C\setminus \{i\}}v_i(j)$.
Note that this implies that the utility for a singleton coalition is~$0$.
Since the valuation functions contain all information for computing utilities, we can also succinctly represent an ASHG as the pair $(N,v)$, where $v = (v_i)_{i\in N}$ is the tuple of valuation functions.
This is equivalent to providing a complete, directed, weighted graph, where the weights of the directed edges induce the valuation functions.

An ASHG $(N,v)$ is said to be \emph{symmetric} if for every pair of distinct agents $i,j\in N$, it holds that $v_i(j) = v_j(i)$.
We write $v(i,j)$ for the symmetric valuation between~$i$ and~$j$.
A complete undirected weighted graph can represent a symmetric ASHG\@.
Furthermore, a subclass of ASHGs is (valuation) \emph{restricted} if it contains only valuations from a fixed set of allowed values.
We will mainly investigate two particular restricted classes, namely AFGs and AEGs.
They restrict the allowed valuations to $\{-1,n\}$ and $\{-n,1\}$ respectively.

Given a game $G = (N,u)$, a partition $\pi$ is said to be \emph{individually rational (IR)} if each agent likes it at least as much as being alone, i.e., for all $i\in N$ it holds that $u_i(\pi)\ge0$.
\begin{definition}
    The \emph{Nash Welfare} of an IR partition $\pi$ is given by the geometric mean of all agents' utilities: $\NW(\pi)=\sqrt[n]{\prod_{i\in N}u_i(\pi)}$.
\end{definition}
We exclusively consider IR partitions when talking about Nash welfare since it is not well-defined if there are agents with negative utility in present.\footnote{\citet{AGG+15b} ignore this issue. Since they use the product formulation of Nash welfare, it is technically still well-defined but the concept makes little sense when multiplying negative utilities.} In general, in any model, working with Nash welfare only makes sense if agents exclusively exhibit non-negative utilities. Using, e.g., the product does not solve the issue since it is completely unclear what it would mean to incorporate negative utilities into the product (multiplying two negative values would lead to the signs being canceled).
Note that the restriction to IR partitions is no issue for the existence of optimal Nash welfare, since, by the all-singletons partition, there always exists at least one IR partition.\footnote{\Cref{sec:bound_number} is an exception.}
We denote by $\pi^*$ a partition with optimal Nash welfare.

There are different possible ways of how to define the Nash welfare which are equivalent up to monotone transformations.
We decided to use the formulation by geometric mean instead of the product (as used, e.g., by \citet{AGG+15b}).
This is mainly due to the fact that writing down meaningful approximation guarantees becomes problematic when all utilities are just multiplied.
Think of a partition in which all agents receive half the utility they get in the one with maximal Nash welfare.
With our definition, this results in an approximation ratio of $2$ while for the product version it would be $2^n$.
Hence, by taking the $n$-th root, we avoid exponential growth in $n$ and instead make approximation ratios better comparable for a varying number of agents.\footnote{In contrast, for utilitarian welfare, it does not influence the approximation ratio whether one uses the sum (as usual) or average of agent utilities.}
Also, the geometric mean version is standard in the literature on approximation algorithms \citep{Lee17a,CoGz18a,GaMu23a,JaVa24a,GSVY25a}.

We cannot assume w.l.o.g.\ that instances are symmetric.
For utilitarian welfare, this is possible because it is invariant under the following transformation.
For any pair of agents $i$ and $j$, consider the original valuations $v_i(j)$ and $v_j(i)$, and substitute $v'(i,j)=\frac{v_i(j)+v_j(i)}2$ for both of them.
A simple example with four agents shows this does not work for Nash welfare.
Let $N=\{a,b,c,d\}$, and $v_a(b)=v_c(d)=4$, $v_b(a)=v_d(c)=0$, $v(a,c)=v(b,d)=1$, $v(a,d)=v(b,c)=-10$.
Here, $\pi=\{\{a,b\},\{c,d\}\}$ leaves $b$ and $d$ with a utility of $0$, turning the Nash welfare $0$ as well.

An algorithm is said to achieve an $\alpha$-approximation (an approximation ratio\slash{}guarantee of $\alpha$) for (a certain class of) ASHGs, if for every instance (of that class), for the partition $\pi$ formed by the algorithm, it holds that $\alpha\cdot\NW(\pi)\ge\NW(\pi^*)$, where $\alpha\ge1$.

Next, we define different types of single-agent deviations and their corresponding stability notions.
Assume, we are given a hedonic game $G=(N,u)$, an agent $i\in N$, and a partition $\pi$ of $N$.
A deviation is a move of some agent from one coalition to another.
Formally, from partition $\pi$, a different partition $\pi'$ can be reached via a deviation of agent $i$ if for all $j \in N \setminus \{i\}$ it holds that $\pi(j) \setminus \{i\} = \pi'(j) \setminus \{i\}$.
A deviation by $i$ is called
\begin{itemize}
    \item a \emph{Nash deviation} if $u_i(\pi')>u_i(\pi)$,
    \item an \emph{individual deviation} if it is a Nash deviation, and for all $j\in\pi'(i)$ it holds that $u_j(\pi')\ge u_j(\pi)$,
    \item a \emph{contractual Nash deviation} if it is a Nash deviation, and for all $j\in\pi(i)$ it holds that $u_j(\pi')\ge u_j(\pi)$,
    \item a \emph{contractual individual deviation} if it is an individual deviation and a contractual Nash deviation.
\end{itemize}
For each of these types of deviations, let a stability notion be defined by the absence of respective deviations.
This way, we get Nash stability (NS), individual stability (IS), contractual Nash stability (CNS), and contractual individual stability (CIS). Note that Nash stability implies the other three, while both contractual Nash stability and individual stability imply contractual individual stability (see \citet{AzSa15a} for a more extensive overview).
In every hedonic game, there exists a CIS partition. This is not the case for the other stability concepts.
For each of them, there even exists an ASHG without any such stable partition \citep{BoJa02a,SuDi07b}.
However, using a potential function argument, \citet{BoJa02a} show that every symmetric ASHG admits a Nash stable partition.

Finally, we introduce a novel restriction for deviations. From a partition $\pi$, we call a deviation (of any type) by agent $i$ \emph{non-abandoning} if $\lvert\pi(i)\rvert\neq2$, that is, $i$ does not leave behind another agent as a singleton. This leads to the concepts of, e.g., non-abandoning Nash deviations and non-abandoning individual deviations which will be used later.

\section{Aversion to Enemies Games}\label{sec:aeg}

In an AEG, each agent classifies each other as being either a friend or an enemy.
Agents have the goal of being in a coalition with as few enemies as possible, and, among two coalitions containing the same number of enemies, prefer the one with more friends in it.
This is modeled as an ASHG by assigning a valuation of $-n$ to enemies and a valuation of $1$ to friends.
For an agent $i$, we denote by $F_i\subseteq N$ the set of friends, and by $E_i\subseteq N$ the set of enemies of $i$ (where of course $F_i\cap E_i=\emptyset$).

AEGs constitute a subclass of ASHGs which naturally captures several scenarios while imposing tight conditions on partitions to be IR.
All members of an agent's coalition have to be her friends.
As a warm-up, we will formalize this easy but useful observation in graph-theoretic terms.
Given an AEG $G=(N,v)$, we denote by
the \emph{mutual-friendship graph} $G_M=(N,E^+\subseteq E)$ the undirected graph which features an edge $(i,j)$ if $v_i(j)=1$ and $v_j(i)=1$.

\begin{restatable}{lemma}{cliqueslemma}\label{lem:cliques}
In an AEG, a partition $\pi$ is IR if and only if every coalition in $\pi$ is a clique in the mutual-friendship graph, meaning no coalition contains a pair of agents such that at least one of them considers the other an enemy.
\end{restatable}

As a result, we can safely disregard one-way friendships when looking for an IR partition and hence when looking for partitions with high Nash welfare. Effectively, we can w.l.o.g.\ restrict attention to symmetric instances.
We use \Cref{lem:cliques} to show computational hardness of the Nash welfare for this restricted class.

\begin{restatable}{theorem}{hardness}\label{thm:hardness}
    Computing a partition with optimal Nash welfare is \NP-hard for symmetric AEGs.
\end{restatable}

\Cref{thm:hardness} motivates us to consider approximation to Nash welfare. But before that, we give a structural insight into how partitions with high Nash welfare look like.
To this end, we study the effect of non-abandoning individual deviations.

\begin{restatable}{theorem}{nonabandon}\label{prop:non-abandon}
    In an AEG, for a partition $\pi$, after any non-abandoning individual deviation resulting in $\pi'$, we have $\NW(\pi')\ge\NW(\pi)$.
    If $\NW(\pi)>0$, then even $\NW(\pi')>\NW(\pi)$.
\end{restatable}

Since we need to maintain a partition into cliques on the mutual-friendship graph in order to keep IR, we see that every non-abandoning individual deviation will be from a coalition $C$ to a coalition $C'$ with $\lvert C\rvert\le\lvert C'\rvert$ (before the deviation).
Thus, the result result shows that partitions which are more ``unbalanced'' with respect to the coalition sizes achieve a higher Nash welfare as long as the number of coalitions is the same (and coalitions are cliques in the mutual-friendship graph).
This is true because in AEGs, as long as we have cliques as coalitions, the utility of each agent solely depends on the size of his coalition, not on the identity of the other members.
\Cref{prop:non-abandon} also enables us to arrive at a locally optimal partition, i.e., one such that no single-agent deviation of any kind improves the welfare, by running the dynamics of non-abandoning individual deviations.
It is open, though, whether this runs in polynomial time.
However, when we are interested in finding a partition with globally good Nash welfare, the deviation dynamics is of limited usefulness.

\begin{example}
In \Cref{fig:example-coalitions}, we see that the optimal partition $\pi^*$ cannot be reached with non-abandoning individual deviations when starting from $\pi$ even though both partitions have the same number of coalitions.
    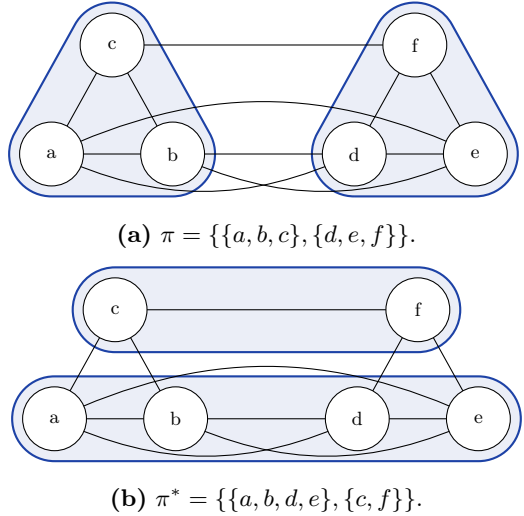
\begin{figure}[H]
    \centering
\begin{tikzpicture}[scale=0.80, every node/.style={transform shape}]
\coordinate (v1) at (0,-0.2);
      \coordinate (v2) at (2,-0.2);
      \coordinate (v3) at (1,1.6);
    
      \coordinate (v4) at (5,-0.2);
      \coordinate (v5) at (7,-0.2);
      \coordinate (v6) at (6,1.6);
    
\draw[thick, CoalitionColor, fill=CoalitionColor!50, fill opacity=0.18]
        \convexpath{v1,v3,v2}{.70cm};
      \draw[thick, CoalitionColor, fill=CoalitionColor!50, fill opacity=0.18]
        \convexpath{v4,v6,v5}{.70cm};

\draw (v1)--(v2)--(v3)--(v1);
      \draw (v4)--(v5)--(v6)--(v4);
    
      \draw [bend right=25] (v1) to (v4);
      \draw [bend right=25](v2) to (v5);
      \draw[bend left=25] (v1) to (v5);
      \draw (v2) to (v4);
      \draw (v3)--(v6);

      \node[agent] at (v1) {a};
      \node[agent] at (v2) {b};
      \node[agent] at (v3) {c};
      \node[agent] at (v4) {d};
      \node[agent] at (v5) {e};
      \node[agent] at (v6) {f};
    \end{tikzpicture}
    
    \smallskip
    \textbf{(a)} $\pi=\{\{a,b,c\},\{d,e,f\}\}$.

\centering
    \medskip
    \begin{tikzpicture}[scale=0.80, every node/.style={transform shape}]
\coordinate (v1) at (0,-0.2);
      \coordinate (v2) at (2,-0.2);
      \coordinate (v3) at (1,1.6);
    
      \coordinate (v4) at (5,-0.2);
      \coordinate (v5) at (7,-0.2);
      \coordinate (v6) at (6,1.6);
    
\draw[thick, CoalitionColor, fill=CoalitionColor!50, fill opacity=0.18]
        \convexpath{v1,v5}{.70cm};
      \draw[thick, CoalitionColor, fill=CoalitionColor!50, fill opacity=0.18]
        \convexpath{v3,v6}{.70cm};

\draw (v1)--(v2)--(v3)--(v1);
      \draw (v4)--(v5)--(v6)--(v4);
    
      \draw [bend right=25] (v1) to (v4);
      \draw [bend right=25](v2) to (v5);
      \draw[bend left=25] (v1) to (v5);
      \draw (v2) to (v4);
      \draw (v3)--(v6);
    
\node[agent] at (v1) {a};
      \node[agent] at (v2) {b};
      \node[agent] at (v3) {c};
      \node[agent] at (v4) {d};
      \node[agent] at (v5) {e};
      \node[agent] at (v6) {f};
    \end{tikzpicture}
    
    \smallskip
    \textbf{(b)} $\pi^*=\{\{a,b,d,e\},\{c,f\}\}$.

    \caption{The mutual-friendship graph of an AEG with two different clique partitions highlighted.
    In $\pi$, no individual deviation is possible. In particular, $\pi^*$ cannot be reached via non-abandoning individual deviations.
}
\label{fig:example-coalitions}
    \end{figure}

\end{example}

For our approximation algorithm, we rely on a very useful extension of graph matching which may include triangles as well.
Given a graph $G$ and a family of graphs $\mathcal{F}$, a \emph{packing} of $G$ with $\mathcal{F}$ is a partition $\pi$ of $G$ such that every subgraph $S\in\pi$ is (isomorphic to) one of the graphs in $\mathcal{F}$.
We refer to the cardinality of a packing as its size.
A packing covering the entire graph is called a \emph{factor} of~$G$.
\Citet{CHP82a} and \citet{HeKi84a} independently showed that it is possible to compute in polynomial time a packing of maximum size if $\mathcal{F}$ includes the $K_2$ (i.e., we can match single edges).
Now, using $\{K_2,K_3\}$-factors of underlying graphs will turn out helpful for AEGs.
It is an idea which has been used in the literature on hedonic games for similar purposes \citep{Bull19a,ABH11c}.
First, we examine a helpful fact about when the optimal Nash welfare is actually $0$.

\begin{restatable}{lemma}{aegzero}\label{lem:aeg_zero}
    In an AEG $G$, the optimal Nash welfare is strictly positive if and only if the mutual-friendship graph $G_M$ admits a factor with $\{K_2,K_3\}$.
\end{restatable}

Not only can we use the algorithm by \citet{HeKi84a} to test whether the optimal Nash welfare is positive, but we can in such cases also compute a partition with an approximation guarantee.

\begin{restatable}{theorem}{aegapprox}\label{thm:aeg_approx}
    For AEGs, there exists an efficient algorithm achieving an $(n-1)$-approximation of Nash welfare.
\end{restatable}

The interesting question is of course whether we can do better.
While we have to ultimately leave this open, we prove first inapproximability results for ASHG subclasses akin to AEGs in \Cref{sec:hardness}.
Moreover, note that due to the inapproximability of finding large cliques in graphs \citep{Zuck07a}, natural limitations with respect to the achievable approximation ratios are to be expected.
It is known, for example, that this is the case for utilitarian welfare maximization in AEGs.
There, it is \NP-hard to approximate the optimum within $\mathcal{O}(n^{1-\varepsilon})$ for any constant $\varepsilon>0$ \citep{FKV22a} while the packing-based algorithm achieves an $\mathcal{O}(n)$-approximation (as in our case).

\section{Appreciation of Friends Games}\label{sec:afg}

In AFGs, just as in AEGs, agents view other agents as either friends or enemies.
Only the order in which they care about them is switched.
Agents have the goal of being in a coalition with as many friends as possible, and, among two coalitions containing the same number of friends, prefer the one with fewer enemies in it.
This is modeled as an ASHG by assigning a valuation of $n$ to friends and a valuation of $-1$ to enemies.
We make use of the same notation as in \Cref{sec:aeg} with $F_i$, $E_i$,
and $G_M$, adapted straightforward to AFGs.

Because of the positive shift in valuations compared to AEGs, it seems easier to satisfy agents in this model.
Indeed, we observe more freedom in constructing IR partitions.
Only in games with agents who absolutely cannot be pleased, we have to give up on Nash welfare as a meaningful measure.

\begin{restatable}{lemma}{afgzero}\label{prop:afg_zero}
In an AFG, the optimal Nash welfare is $0$ if and only if there exists an agent who considers no other agent a friend.
\end{restatable}

This observation shows that an interesting sensible refinement of Nash welfare would be one where agents without any friends are completely disregarded in the calculation.
However, we do not further investigate this idea here and leave it open for future work.

We do not show computational hardness of the optimal Nash welfare for AFGs.
However, we remark that also for all other well-established welfare notions, to the best of our knowledge, the complexity is open for this class of games.
The already mentioned large range of reasonable partitions seems to make it rather difficult to pinpoint what exactly turns computing optimal partitions hard (if anything).
Therefore, let us discuss two simple ideas that seem reasonable at first glance but do not (directly) lead to good approximations of Nash welfare.
The first idea is simply forming the grand coalition.
The rationale behind it is that friendship relations have significantly more impact than enemies.
This approach already seems questionable in the cases where there is an agent disliking everyone else.
On the one hand, the optimal Nash welfare in such a situation is $0$ anyways, on the other hand, putting the agent who can't be pleased inside the same coalition as the others makes the outcome in some sense worse than needed.
So assume, we are not in this degenerate scenario characterized by \Cref{prop:afg_zero}, i.e., each agent considers at least one other a friend.
Still, it is possible that many agents only like one other. It gets complex then to identify the worst case instances regarding the ratio between the optimal partition and the grand coalition.
The second idea is to output a partition into cliques of the mutual-friendship graph.
This is the basis of our idea, however, if not refined it fails because it may leave agents in singletons, resulting in a Nash welfare of $0$. Moreover, in order to guarantee efficient computability, caution has to be taken regarding what kind of partition into cliques is desired.

We show that, similar to AEGs, a packing-based algorithm gives rise to an approximation of Nash welfare.
This time, however, it needs to be completed to an IR partition after computing the initial maximum packing.
However, with this technique alone, merely an approximation ratio of $\mathcal{O}(n^2)$ is achieved.
The approximation ratio is worse here because, in contrast to AEGs, we have to actually deal with games where the maximum size $\{K_2,K_3\}$-packing is not a factor of the friendship graph.
Comparing \Cref{lem:aeg_zero} and \Cref{prop:afg_zero} demonstrates that, for AFGs, we have significantly more possible partitions to consider.
However, we can exploit the flexibility of IR partitions to obtain a more sophisticated algorithm which makes use of certain non-abandoning Nash deviations.

\begin{algorithm}[H]
\caption{Packing-based partition with deviation refinement for symmetric AFGs}\label{alg:matching-congestion}
\begin{algorithmic}[1]
  \Require Symmetric AFG on agents $N$ and mutual-friendship graph $G_M = (N, E^+)$
  \Ensure Partition $\pi$ which is a $2n$-approximation of Nash welfare
  \Statex \textbf{Phase 1: Build initial coalitions}
  \If{$G_M$ contains an isolated vertex}
    \State \Return all-singletons partition
  \EndIf
  \State $\pi\gets$ a maximum size $\{K_2,K_3\}$-packing in $G_M$; unpacked agents are in singleton coalitions
  \Comment maximum size refers to the number of packed agents
  \State $\hat{U}\gets$ unpacked agents
  \State $U\gets$ unpacked agents
  \ForAll{$C\in\pi$ with $\lvert C\rvert=2$}
\If{there is exactly one $i\in C$ who has no friend in $\hat{U}$}
        \State $U\gets U\cup\{i\}$
      \EndIf
\EndFor
\Statex
  \Statex \textbf{Phase 2: Deviation Dynamics}
  \While{there exists an agent $i\in U$ who can perform a non-abandoning Nash deviation from $C$ to $C'$ while for all $j\in C\setminus\{i\}\colon u_j(C\setminus\{i\})>0$}
    \Comment $i$'s deviation must not leave behind others with $0$ utility
\State $\pi\gets\pi\setminus\{C,C'\}$
    \State $\pi\gets\pi\cup\{C'\cup\{i\}\}$
    \If{$C\setminus \{i\}\neq\emptyset$}
      \State $\pi\gets\pi\cup\{C\setminus \{i\}\}$
    \EndIf
    \Comment perform the deviation
  \EndWhile
  \State \Return $\pi$
\end{algorithmic}
\end{algorithm}

\begin{restatable}{theorem}{afgapprox}
    For symmetric AFGs, \Cref{alg:matching-congestion} achieves a $2n$-approximation of Nash welfare and runs in polynomial time.
\end{restatable}

From the analysis, we see that our algorithm performs well on star-like structures in the mutual-friendship graph but has stronger limitations in the presence of large cliques, as can be expected.

\section{Hardness of Approximation}\label{sec:hardness}

Now that we have seen how to approximate Nash welfare in structured subclasses of ASHGs, we approach the problem from the other side and show that for general ASHGs even obtaining a \PTAS{} is unlikely to be possible.
The following result employs a reduction by \citet{GSVY25a} who in turn base their hardness result on the work of \citet{GaMu23a}.
Note that our result is not implied by existing hardness of approximation proofs in the fair division literature (e.g., \citet{Lee17a}) because the agent set is different.

\begin{restatable}{theorem}{approxhardness}
    It is \NP-hard to approximate Nash welfare for ASHGs within a factor smaller than \(1.0000759\), even if all valuations are in $\{-H,0,1,2\}$ where $H$ is some penalty valuation larger than $2n$.
\end{restatable}

The roles of $H$ is enforcing agents to be in different coalitions. Hence, it has to be sufficiently large.
We see that the class of games for which we obtain hardness of approximation is quite similar to AEGs, mainly differing in the fact that friendship valuations can be $1$ or $2$, and neutrals exist.
It remains an interesting open question whether a hardness can also be shown for proper AEGs.

\section{Restricting the Feasible Partitions}\label{sec:restrict}

We again turn our attention to AEGs.
For different scenarios in which scarcity of resources plays a role, imposing bounds on the feasible partitions leads to more realistic models.

\subsection{Bounding the Number of Coalitions}\label{sec:bound_number}

In this section, we require an upper bound $k$ on the number of allowed coalitions.
Without any further restrictions, this might give rise to instances in which no IR partition exists, and hence the optimal Nash welfare is undefined \footnote{One could even argue that IR itself becomes a problematic definition when it is impossible to put all agents into singleton coalitions. However, we do not further discuss this here.}.
Because of this, we explicitly disallow instances with this property in further considerations. First, we investigate the case of $k=2$.

\begin{restatable}{theorem}{twocoalitions}
    In AEGs, if the number of allowed coalitions is at most $2$, a partition of maximal Nash welfare can be computed in polynomial time. \end{restatable}

In contrast, for $k\ge3$, we get a very strong negative result.
It is computationally hard to approximate the Nash welfare within any finite ratio.
This means, even poor approximation guarantees like, e.g., $\mathcal{O}(2^n)$ are out of the question.

\begin{restatable}{theorem}{hardnessnumcoal}
    Fix any $k\ge 3$.
    In an AEG, it is \NP-complete to decide whether there exists a partition into at most $k$ coalitions with strictly positive Nash welfare.
\end{restatable}

\subsection{Bounding the Size of Coalitions}\label{sec:bound_size}

In this section, we impose an upper bound $s$ on the size of allowed coalitions. For $s=2$, we are essentially in the realm of matching on general graphs. Thus, we can efficiently maximize Nash welfare.
\begin{restatable}{proposition}{algosizecoal}
    For ASHGs, if the allowed size of coalitions is upper bounded by $2$, the optimal Nash welfare can be computed in polynomial time.
\end{restatable}

For $s\ge3$, we get computational hardness for several classes of ASHGs, as described by the following result.
\begin{restatable}{theorem}{hardnesscoalsize}
    Let $\mathcal{A} \subseteq \mathbb{Q}$ be a finite set of allowed valuations with $|\mathcal{A}|\ge 2$ and $\max \mathcal{A} > 0$.
    Consider the class of symmetric ASHGs in which, for every pair of distinct agents $i\neq j$, we have $v(i,j)\in\mathcal{A}$, and where every value in $\mathcal{A}$ is allowed to occur arbitrarily often.

    Fix any integer $s\ge 3$.
    It is \NP-hard to compute a partition with optimal Nash welfare among all partitions whose coalitions have size at most $s$, even when restricted to games in this class.
\end{restatable}

We formulate the following corollary because of the focus on AEGs and AFGs in this paper but as the theorem shows, the hardness extends to other well-known classes like, e.g., friends and enemies games or AEGs with neutrals, which we do not define here.

\begin{corollary}
    Fix any integer $s\ge 3$.
    In symmetric AEGs and symmetric AFGs, it is \NP-hard to compute a partition with optimal Nash welfare among all partitions whose coalitions have size at most $s$.
\end{corollary}

Next, we describe how to obtain an efficient approximation of Nash welfare, where the guarantee depends on $s$.
Run an algorithm for maximum $\{K_2,K_3\}$-packing and check whether it returns a factor.
If not, return the all-singletons partition.
Since in this case the best achievable welfare is $0$ as well (\Cref{lem:aeg_zero}), the output is optimal.
If it returns a factor, return it as the partition.

\begin{restatable}{theorem}{aegapproxcoalsize}
    For AEGs and any fixed $s\ge 3$, the procedure above is a $(s-1)$-approximation of Nash welfare.
\end{restatable}

\section{Stability Guarantees}\label{sec:stability}

In previous sections, we have investigated the possibilities of computing or approximating Nash welfare. Here, we complement the picture and turn to questions of stability.
In many scenarios, it is not only interesting whether a partition is fair and efficient but also whether agents will stick to it.
We compare our findings to the existing ones for utilitarian welfare.
The definitions of the different stability concepts under consideration here can be found in \Cref{sec:prelims}.

\begin{restatable}{theorem}{nashcns}\label{thm:nash_cns}
    In a symmetric ASHG, if $\NW(\pi)$ is a finite factor approximation of $\NW(\pi^*)$, and $\NW(\pi^*)>0$, then $\pi$ is CNS.
\end{restatable}

\Cref{thm:nash_cns} shows that even approximations of Nash welfare satisfy this natural stability criterion for all ASHGs, as long as they are symmetric, independent of the exact approximation ratio.
The key is that all agents must receive positive utility, and this is actually a sufficient condition for the theorem to hold.
Maximizing utilitarian welfare even satisfies the stronger stability notion of NS in symmetric instances \citep{BoJa02a} but only if we actually get an optimal partition.
For approximations of it, no imposed stability guarantees are known or immediately deducible.
Maximizing Nash welfare, in contrast, fails to imply NS.

\begin{remark}
    A partition that maximizes the Nash welfare is not necessarily IS, even if the game is symmetric.
    Suppose we are in the symmetric AEG $G$ depicted in \Cref{fig:is_deviation} by its mutual friendship graph $G_M$.
    An IR partition of $G$ is a partition into cliques of $G_M$.
    \iffalse \begin{figure}[H]
    \centering
    \begin{tikzpicture}[scale=0.80, every node/.style={transform shape}, poslab/.style={sloped, midway, fill=white, inner sep=1pt},
        neg/.style={red, dashed}]

    \coordinate (v1) at (0,0);
    \coordinate (v2) at (2,1);
    \coordinate (v3) at (0,2);
    \coordinate (v4) at (4,1);

    \draw (v1) -- (v2) -- (v3) -- (v1);
    \draw (v2) -- (v4);

    \node[agent] at (v1) {a};
    \node[agent] at (v2) {b};
    \node[agent] at (v3) {c};
    \node[agent] at (v4) {d};
    
    \end{tikzpicture}
    \end{figure}
    \fi 

    The optimal partition is $\{\{a,c\},\{b,d\}\}$.
    However, agent~$b$ can individually deviate to the coalition with agents~$a$ and~$c$.
    The resulting partition $\{\{a,b,c\},\{d\}\}$ yields a Nash welfare of $0$, since agent~$d$ becomes a singleton.
    This shows that IS is violated, and as a consequence, the stronger notion of NS is violated too.

    \begin{figure}[H]
    \centering
    
    \begin{minipage}{0.45\linewidth}
    \centering
    \begin{tikzpicture}[scale=0.80, every node/.style={transform shape}, poslab/.style={sloped, midway, fill=white, inner sep=1pt},
        neg/.style={red, dashed}]
    \coordinate (v1) at (0,0);
    \coordinate (v2) at (2,1);
    \coordinate (v3) at (0,2);
    \coordinate (v4) at (4,1);

    \draw[thick, CoalitionColor, fill=CoalitionColor!50, fill opacity=0.18]
        \convexpath{v1,v3}{.70cm};
      \draw[thick, CoalitionColor, fill=CoalitionColor!50, fill opacity=0.18]
        \convexpath{v2,v4}{.70cm};

    \draw (v1) -- (v2) -- (v3) -- (v1);
    \draw (v2) -- (v4);

    \node[agent] at (v1) {a};
    \node[agent] at (v2) {b};
    \node[agent] at (v3) {c};
    \node[agent] at (v4) {d};

    \end{tikzpicture}

    \smallskip
    \textbf{(a)} $\pi^*=\{\{a,c\},\{b,d\}\}$.
    \end{minipage}
    \hfill
    \begin{minipage}{0.45\linewidth}
    \centering
    \begin{tikzpicture}[scale=0.80, every node/.style={transform shape}, poslab/.style={sloped, midway, fill=white, inner sep=1pt},
        neg/.style={red, dashed}]
        \coordinate (v1) at (0,0);
    \coordinate (v2) at (2,1);
    \coordinate (v3) at (0,2);
    \coordinate (v4) at (4,1);

    \draw[thick, CoalitionColor, fill=CoalitionColor!50, fill opacity=0.18]
        \convexpath{v1,v3,v2}{.70cm};
      \path[draw=CoalitionColor, line width=1.1pt, fill=CoalitionColor!10]
  (v4) circle[radius=0.70cm];

    \draw (v1) -- (v2) -- (v3) -- (v1);
    \draw (v2) -- (v4);

    \node[agent] at (v1) {a};
    \node[agent] at (v2) {b};
    \node[agent] at (v3) {c};
    \node[agent] at (v4) {d};

    \end{tikzpicture}

    \smallskip
    \textbf{(b)} $\pi=\{\{a,b,c\},\{d\}\}$.
    \end{minipage}

    \caption{The partition of optimal Nash welfare (a) is prone to an individual deviation of agent $b$ (b)}\label{fig:is_deviation}
    \end{figure}
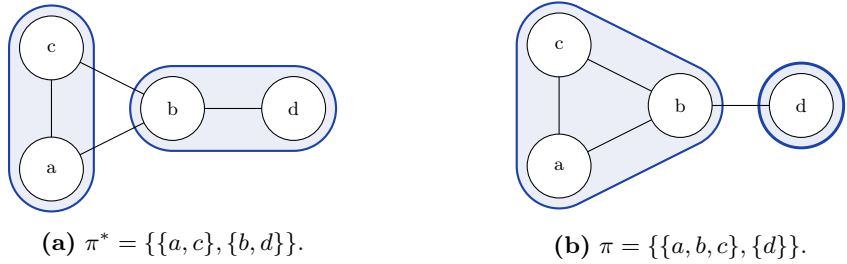
\end{remark}

\section{Conclusion}

We have initiated the study of Nash welfare in ASHGs.
The basis of our computational investigation has been laid by proving \NP-hardness of computing a partition of optimal Nash welfare. For the subclasses of AEGs and AFGs, we obtained approximation algorithms with guarantees of $n-1$ and $2n$, respectively.
It remains an open question, whether these ratios are tight.
While the hardness of approximation of finding the largest clique in graphs \citep{Zuck07a} makes it seem unlikely to make meaningful improvements, for our proof of hardness of approximation, we had to resort to another kind of problem to reduce from. We showed that achieving a $1.0000759$-approximation is \NP-hard by carving out a close relation of ASHGs to two-sided matching.
Moreover, we provided helpful lemmas characterizing partitions with positive and optimal Nash welfare.

We further considered restricting the set of feasible partitions by an upper bound on the number or size of partitions.
Here, we observed the typical dichotomy in computational complexity theory which often occurs across different classes of problems when there is a parameter which can assume values in the natural numbers (in our case the bound on coalition number of size).
For a value of $2$, computing an (optimal) solution is possible in polynomial time while, for values of at least $3$, we get \NP-hardness or even \NP-hardness of approximation.
There exist similar findings for, e.g., the boolean satisfiability problem (2SAT is in \P, 3SAT is \NP-complete), or matching (2-dimensional is in \P, 3-dimensional is \NP-complete).

Finally, we have shifted attention to stability notions implied by partitions with good Nash welfare and set them in relation with known results for utilitarian welfare.
Since the computational hardness of both Nash and utilitarian welfare casts doubt on the practical usefulness of imposed stability guarantees of optimal partitions, we deem it quite attractive that even approximating Nash welfare to any factor yields a reasonable stability guarantee.

Beyond our definition of Nash welfare, as mentioned after \Cref{prop:afg_zero}, a natural refinement for AFGs (and other classes of games) would be to ignore agents who do not exhibit positive valuations to any other agent in the calculation.
This would leave the ordinary Nash welfare untouched whenever it can be positive but provide a way to discriminate between partitions in the degenerate cases where the ordinary Nash welfare can only be $0$.
A similar definition already exists for Nash welfare in fair division \citep{CKM+19a}.

Furthermore, it would be interesting to study Nash welfare in FHGs and MFHGs, or in online models of hedonic games, analyzing the competitive ratio.

\bibliographystyle{splncs04nat}

\begin{thebibliography}{52}
\providecommand{\natexlab}[1]{#1}
\providecommand{\url}[1]{\texttt{#1}}
\providecommand{\urlprefix}{URL }
\expandafter\ifx\csname urlstyle\endcsname\relax
  \providecommand{\doi}[1]{doi:\discretionary{}{}{}#1}\else
  \providecommand{\doi}{doi:\discretionary{}{}{}\begingroup
  \urlstyle{rm}\Url}\fi

\bibitem[{Amanatidis et~al.(2021)Amanatidis, Birmpas, Filos-Ratsikas,
  Hollender, and Voudouris}]{ABF+21a}
Amanatidis, G., Birmpas, G., Filos-Ratsikas, A., Hollender, A., Voudouris,
  A.A.: Maximum {N}ash welfare and other stories about {EFX}. Theoretical
  Computer Science \textbf{863}, 69--85 (2021)

\bibitem[{Axelrod et~al.(1995)Axelrod, Mitchell, Thomas, Bennett, and
  Bruderer}]{AMT+95a}
Axelrod, R., Mitchell, W., Thomas, R.E., Bennett, D.S., Bruderer, E.: Coalition
  formation in standard-setting alliances. Management Science \textbf{41}(9),
  1493--1508 (1995)

\bibitem[{Aziz et~al.(2019)Aziz, Brandl, Brandt, Harrenstein, Olsen, and
  Peters}]{ABB+17a}
Aziz, H., Brandl, F., Brandt, F., Harrenstein, P., Olsen, M., Peters, D.:
  Fractional hedonic games. ACM Transactions on Economics and Computation
  \textbf{7}(2), 1--29 (2019)

\bibitem[{Aziz et~al.(2013{\natexlab{a}})Aziz, Brandt, and
  Harrenstein}]{ABH11c}
Aziz, H., Brandt, F., Harrenstein, P.: Pareto optimality in coalition
  formation. Games and Economic Behavior \textbf{82}, 562--581
  (2013{\natexlab{a}})

\bibitem[{Aziz et~al.(2011)Aziz, Brandt, and Seedig}]{ABS11a}
Aziz, H., Brandt, F., Seedig, H.G.: Optimal partitions in additively separable
  hedonic games. In: Proceedings of the 22nd International Joint Conference on
  Artificial Intelligence (IJCAI), pp. 43--48 (2011)

\bibitem[{Aziz et~al.(2013{\natexlab{b}})Aziz, Brandt, and Seedig}]{ABS11c}
Aziz, H., Brandt, F., Seedig, H.G.: Computing desirable partitions in
  additively separable hedonic games. Artificial Intelligence \textbf{195},
  316--334 (2013{\natexlab{b}})

\bibitem[{Aziz et~al.(2015)Aziz, Gaspers, Gudmundsson, Mestre, and
  T{\"a}ubig}]{AGG+15b}
Aziz, H., Gaspers, S., Gudmundsson, J., Mestre, J., T{\"a}ubig, H.: Welfare
  maximization in fractional hedonic games. In: Proceedings of the 24th
  International Joint Conference on Artificial Intelligence (IJCAI), pp.
  461--467 (2015)

\bibitem[{Aziz and Savani(2016)}]{AzSa15a}
Aziz, H., Savani, R.: Hedonic games. In: Brandt, F., Conitzer, V., Endriss, U.,
  Lang, J., Procaccia, A.D. (eds.) Handbook of Computational Social Choice,
  chap.~15, Cambridge University Press (2016)

\bibitem[{Ballester(2004)}]{Ball04a}
Ballester, C.: {NP}-completeness in hedonic games. Games and Economic Behavior
  \textbf{49}(1), 1--30 (2004)

\bibitem[{Banerjee et~al.(2001)Banerjee, Konishi, and S{\"o}nmez}]{BKS01a}
Banerjee, S., Konishi, H., S{\"o}nmez, T.: Core in a simple coalition formation
  game. Social Choice and Welfare \textbf{18}, 135--153 (2001)

\bibitem[{Bogomolnaia and Jackson(2002)}]{BoJa02a}
Bogomolnaia, A., Jackson, M.O.: The stability of hedonic coalition structures.
  Games and Economic Behavior \textbf{38}(2), 201--230 (2002)

\bibitem[{Brandl et~al.(2022)Brandl, Brandt, Greger, Peters, Stricker, and
  Suksompong}]{BBP+19a}
Brandl, F., Brandt, F., Greger, M., Peters, D., Stricker, C., Suksompong, W.:
  Funding public projects: {A} case for the {N}ash product rule. Journal of
  Mathematical Economics \textbf{99}, 102585 (2022)

\bibitem[{Bullinger(2020)}]{Bull19a}
Bullinger, M.: Pareto-optimality in cardinal hedonic games. In: Proceedings of
  the 19th International Conference on Autonomous Agents and Multiagent Systems
  (AAMAS), pp. 213--221 (2020)

\bibitem[{Bullinger et~al.(2025)Bullinger, Chatziafratis, and Shahkar}]{BCS25a}
Bullinger, M., Chatziafratis, V., Shahkar, P.: Welfare approximation in
  additively separable hedonic games. In: Proceedings of the 24th International
  Conference on Autonomous Agents and Multiagent Systems (AAMAS), pp. 418--426
  (2025)

\bibitem[{Bullinger et~al.(2026{\natexlab{a}})Bullinger, Dunajski, Elkind, and
  Gilboa}]{BDEG26a}
Bullinger, M., Dunajski, A., Elkind, E., Gilboa, M.: Single-deviation stability
  in additively separable hedonic games with constrained coalition sizes
  (extended abstract). In: Proceedings of the 25th International Conference on
  Autonomous Agents and Multiagent Systems (AAMAS) (2026{\natexlab{a}})

\bibitem[{Bullinger et~al.(2024)Bullinger, Elkind, and Rothe}]{BER24a}
Bullinger, M., Elkind, E., Rothe, J.: Cooperative game theory. In: Rothe, J.
  (ed.) Economics and Computation: An Introduction to Algorithmic Game Theory,
  Computational Social Choice, and Fair Division, chap.~3, pp. 139--229,
  Springer (2024)

\bibitem[{Bullinger and Romen(2025)}]{BuRo25b}
Bullinger, M., Romen, R.: Online coalition formation under random arrival or
  coalition dissolution. ACM Transactions on Algorithms \textbf{22}(1), 4:1--43
  (2025)

\bibitem[{Bullinger et~al.(2026{\natexlab{b}})Bullinger, Romen, and
  Schlenga}]{BRS26a}
Bullinger, M., Romen, R., Schlenga, A.: The power of matching for online
  fractional hedonic games. In: Proceedings of the 37th Annual ACM-SIAM
  Symposium on Discrete Algorithms (SODA), pp. 2163--2194 (2026{\natexlab{b}})

\bibitem[{Caragiannis et~al.(2019)Caragiannis, Kurokawa, Moulin, Procaccia,
  Shah, and Wang}]{CKM+19a}
Caragiannis, I., Kurokawa, D., Moulin, H., Procaccia, A.D., Shah, N., Wang, J.:
  The unreasonable fairness of maximum {N}ash welfare. ACM Transactions on
  Economics and Computation \textbf{7}(3), 12:1--12:32 (2019)

\bibitem[{Cechl{\'a}rov{\'a} and Hajdukov{\'a}(2002)}]{CeHa02a}
Cechl{\'a}rov{\'a}, K., Hajdukov{\'a}, J.: {Computational complexity of stable
  partitions with B-preferences}. International Journal of Game Theory
  \textbf{31}(3), 353--354 (2002)

\bibitem[{Cohen and Agmon(2025)}]{CoAg25a}
Cohen, S., Agmon, N.: Egalitarianism in online coalition formation (extended
  abstract). In: Proceedings of the 24th International Conference on Autonomous
  Agents and Multiagent Systems (AAMAS), pp. 2475--2477 (2025)

\bibitem[{Colbourn(1984)}]{Colb84a}
Colbourn, C.J.: The complexity of completing partial latin squares. Discrete
  Applied Mathematics \textbf{8}(1), 25--30 (1984)

\bibitem[{Cole and Gkatzelis(2018)}]{CoGz18a}
Cole, R., Gkatzelis, V.: Approximating the {N}ash social welfare with
  indivisible items. {SIAM} Journal on Computing \textbf{47}(3), 1211--1236
  (2018)

\bibitem[{Cornu{\'e}jols et~al.(1982)Cornu{\'e}jols, Hartvigsen, and
  Pulleyblank}]{CHP82a}
Cornu{\'e}jols, G., Hartvigsen, D., Pulleyblank, W.R.: Packing subgraphs in a
  graph. Operations Research Letters \textbf{1}(4), 139--143 (1982)

\bibitem[{Darmann and Schauer(2015)}]{DaSc15a}
Darmann, A., Schauer, J.: Maximizing {N}ash product social welfare in
  allocating indivisible goods. European Journal of Operational Research
  \textbf{247}(2), 548--559 (2015)

\bibitem[{Dimitrov et~al.(2006)Dimitrov, Borm, Hendrickx, and Sung}]{DBHS06a}
Dimitrov, D., Borm, P., Hendrickx, R., Sung, S.C.: Simple priorities and core
  stability in hedonic games. Social Choice and Welfare \textbf{26}(2),
  421--433 (2006)

\bibitem[{Dr{\`e}ze and Greenberg(1980)}]{DrGr80a}
Dr{\`e}ze, J.H., Greenberg, J.: Hedonic coalitions: Optimality and stability.
  Econometrica \textbf{48}(4), 987--1003 (1980)

\bibitem[{Edmonds(1965)}]{Edmo65a}
Edmonds, J.: Paths, trees and flowers. Canadian Journal of Mathematics
  \textbf{17}, 449--467 (1965)

\bibitem[{Elkind et~al.(2020)Elkind, Fanelli, and Flammini}]{EFF20a}
Elkind, E., Fanelli, A., Flammini, M.: Price of pareto optimality in hedonic
  games. Artificial Intelligence \textbf{288}, 103357 (2020)

\bibitem[{Fioravantes et~al.(2025)Fioravantes, Gahlawat, and
  Melissinos}]{FGM25a}
Fioravantes, F., Gahlawat, H., Melissinos, N.: Exact algorithms and lower
  bounds for forming coalitions of constrained maximum size. In: Proceedings of
  the 39th AAAI Conference on Artificial Intelligence (AAAI), pp. 13847--13855
  (2025)

\bibitem[{Flammini et~al.(2022)Flammini, Kodric, and Varricchio}]{FKV22a}
Flammini, M., Kodric, B., Varricchio, G.: Strategyproof mechanisms for friends
  and enemies games. Artificial Intelligence \textbf{302}, 103610 (2022)

\bibitem[{Flammini et~al.(2021)Flammini, Monaco, Moscardelli, Shalom, and
  Zaks}]{FMM+21a}
Flammini, M., Monaco, G., Moscardelli, L., Shalom, M., Zaks, S.: On the online
  coalition structure generation problem. Journal of Artificial Intelligence
  Research \textbf{72}, 1215--1250 (2021)

\bibitem[{Gabow and Tarjan(1991)}]{GaTa91a}
Gabow, H.N., Tarjan, R.: Faster scaling algorithms for general graph matching
  problems. Journal of the ACM \textbf{38}(4), 815--853 (1991)

\bibitem[{Garey and Johnson(1979)}]{GaJo79a}
Garey, M.R., Johnson, D.S.: Computers and Intractability: A Guide to the Theory
  of NP-Completeness. W. H. Freeman (1979)

\bibitem[{Garg and Murhekar(2023)}]{GaMu23a}
Garg, J., Murhekar, A.: Computing fair and efficient allocations with few
  utility values. Theoretical Computer Science \textbf{962}, 113932 (2023)

\bibitem[{Gokhale et~al.(2025)Gokhale, Sagar, Vaish, and Yadav}]{GSVY25a}
Gokhale, S., Sagar, H., Vaish, R., Yadav, J.: Approximating one-sided and
  two-sided nash social welfare with capacities. In: Proceedings of the 24th
  International Conference on Autonomous Agents and Multiagent Systems (AAMAS),
  pp. 914--922 (2025)

\bibitem[{Harsanyi(1959)}]{Hars59a}
Harsanyi, J.C.: A bargaining model for the cooperative n-person game. In:
  Tucker, A.W., Luce, R.D. (eds.) Contributions to the Theory of Games,
  vol.~IV, pp. 325--355, Princeton University Press (1959)

\bibitem[{Hell and Kirkpatrick(1984)}]{HeKi84a}
Hell, P., Kirkpatrick, D.G.: Packings by cliques and by finite families of
  graphs. Discrete Mathematics \textbf{49}(1), 45--59 (1984)

\bibitem[{Jain and Vaish(2024)}]{JaVa24a}
Jain, P., Vaish, R.: Maximizing {N}ash social welfare under two-sided
  preferences. In: Proceedings of the 38th AAAI Conference on Artificial
  Intelligence (AAAI), pp. 9798--9806 (2024)

\bibitem[{Kaneko and Nakamura(1979)}]{KaNa79a}
Kaneko, M., Nakamura, K.: The nash social welfare function. Econometrica
  \textbf{47}(2), 423--435 (1979)

\bibitem[{Lee(2017)}]{Lee17a}
Lee, E.: {APX}-hardness of maximizing nash social welfare with indivisible
  items. Information Processing Letters \textbf{122}, 17--20 (2017)

\bibitem[{Levinger et~al.(2024)Levinger, Hazon, Simola, and Azaria}]{LHSA24a}
Levinger, C., Hazon, N., Simola, S., Azaria, A.: Coalition formation with
  bounded coalition size. In: Proceedings of the 23rd International Conference
  on Autonomous Agents and Multiagent Systems (AAMAS), pp. 1119--1127 (2024)

\bibitem[{Nash(1950)}]{Nash50b}
Nash, J.F.: The bargaining problem. Econometrica \textbf{18}(2), 155--162
  (1950)

\bibitem[{von Neumann and Morgenstern(1953)}]{vNM53a}
von Neumann, J., Morgenstern, O.: Theory of Games and Economic Behavior.
  Princeton University Press, 3rd edn. (1953)

\bibitem[{Olsen(2012)}]{Olse12a}
Olsen, M.: On defining and computing communities. In: Proceedings of the 18th
  Computing: The Australasian Theory Symposium (CATS), Conferences in Research
  and Practice in Information Technology (CRPIT), vol. 128, pp. 97--102 (2012)

\bibitem[{Peters and Elkind(2015)}]{PeEl15a}
Peters, D., Elkind, E.: Simple causes of complexity in hedonic games. In:
  Proceedings of the 25th International Joint Conference on Artificial
  Intelligence (IJCAI), pp. 617--623 (2015)

\bibitem[{Ray(2007)}]{Ray07a}
Ray, D.: A Game-Theoretic Perspective on Coalition Formation. Oxford University
  Press (2007)

\bibitem[{Suksompong(2023)}]{Suks23a}
Suksompong, W.: A characterization of maximum {N}ash welfare for indivisible
  goods. Economics Letters \textbf{222}, 110956 (2023)

\bibitem[{Sung and Dimitrov(2007)}]{SuDi07b}
Sung, S.C., Dimitrov, D.: On myopic stability concepts for hedonic games.
  Theory and Decision \textbf{62}(1), 31--45 (2007)

\bibitem[{Sung and Dimitrov(2010)}]{SuDi10a}
Sung, S.C., Dimitrov, D.: Computational complexity in additive hedonic games.
  European Journal of Operational Research \textbf{203}(3), 635--639 (2010)

\bibitem[{Zech and Bullinger(2026)}]{ZeBu26a}
Zech, V., Bullinger, M.: Deviation dynamics in cardinal hedonic games. In:
  Proceedings of the 40th AAAI Conference on Artificial Intelligence (AAAI)
  (2026), forthcoming

\bibitem[{Zuckerman(2007)}]{Zuck07a}
Zuckerman, D.: Linear degree extractors and the inapproximability of max clique
  and chromatic number. Theory of Computing \textbf{3}, 103--128 (2007)

\end{thebibliography}

\clearpage

\appendix

\section*{Appendix}\label{sec:appendix}

In the appendix, we present the proofs omitted in the main body of the paper.

\subsection*{Proofs for Section \ref{sec:aeg}}

Here, we provide proofs related to AEGs.

\cliqueslemma*
\begin{proof}
Given the mutual-friendship graph $G_M$ of an AEG, let $\pi$ be an IR partition and $C=\pi(i)$ be the coalition of some player $i$.
Suppose for contradiction that $C$ contains two players $i$ and $j$ who are not mutual friends, w.l.o.g.\ $v_i(j)=-n$.
Then,
\begin{align*}
    u_\pi(i)&=\sum_{k\in C\setminus\{i\}} v_i(k)
    =v_i(j)+\sum_{k\in C\setminus\{i,j\}} v_i(k)\\
    &\le-n + (|C|-2)\cdot 1
    \le-n + (n-2) < 0\text,
\end{align*}
thus violating IR.

Conversely, suppose every coalition in $\pi$ induces a symmetric clique in $G_M$.
Fix any player $i$ and let $C=\pi(i)$.
\[
u_\pi(i)=\sum_{k\in S\setminus\{i\}} v_i(k)=|C|-1\ge 0
\]
So $\pi$ is IR.
\end{proof}

\hardness*
We reduce from the following problem:
\begin{problemdef}
    \problemtitle{\textsc{PartitionIntoTriangles}}
    \probleminput{An undirected graph $G=(V,E)$, with $|V|=3q=n$ for some integer $q$.}
    \problemsolution{`Yes' if the vertices of $G$ can be partitioned in $q$ disjoint sets $V_1,...,V_q$, each containing exactly $3$ vertices, such that each of these $V_i$ is the vertex set of a triangle in $G$. `No' otherwise.}
\end{problemdef}
\begin{proof}

It is known that \textsc{PartitionIntoTriangles} is \NP-complete even for tripartite graphs in which each of the independent three sets is specified \citep[Theorem 2.1]{Colb84a}.
Given a tripartite graph $G$ as an instance of \textsc{PartitionIntoTriangles}, we construct an AEG $(N,v)$ with $N=V$ and symmetric preferences of the players as follows: $v_i(j)=1$ if $(i,j) \in E$, $v_i(j)=-n$ otherwise.

We claim that a partition into triangles $\pi^*$ is the unique partition with the maximum Nash welfare if it exists. In such a partition, each agent has two friends in her coalition and hence gets a utility of $2$. The Nash welfare of $\pi^*$ is therefore $2$.

Now, assume there exists an IR partition $\mu$ other than a partition into triangles which achieves Nash welfare of at least $2$. If all coalitions in $\mu$ had size $3$, it would contain at least one coalition which is no clique in the mutual friendship graph and hence would violate IR by \Cref{lem:cliques}. Hence, this cannot be the case. Further, $\mu$ must contain at least one coalition $C$ in which some agent $i$ has a utility greater than $2$ because otherwise there would be agents with utility strictly less than $2$ while all utilities are also upper bounded by $2$, resulting in a Nash welfare less than $2$. In $C$, agent $i$ has at least $3$ friends. But since the graph from which we created the AEG instance is tripartite, $C$ cannot be a clique in the mutual friendship graph, again violating IR by \Cref{lem:cliques}.

We have shown that if $G$ admits a partition into triangles, the corresponding AEG partition achieves a Nash welfare of $2$ in the derived instance, and if $G$ does not admit a partition into triangles, there is no means of achieving a Nash welfare of at least $2$ in the derived AEG instance. Hence, we can decide \textsc{PartitionIntoTriangles} by computing the maximal Nash welfare in a symmmetric AEG, rendering the latter problem \NP-hard.
\end{proof}

\nonabandon*
\begin{proof}
First, consider the case where the deviation is from a singleton coalition. Here, the Nash welfare clearly increases (strictly if no agents with $0$ utility remain). Hence, assume from now on that the deviation is from a coalition of size at least $3$ (remember it has to be non-abandoning). Moreover, assume that the initial Nash welfare is positive (otherwise it may just stay $0$).
Let $\pi=\{C_i\mid 1\le i\le k\}$ denote the starting partition (into cliques $C_i$) when the Nash welfare before a non-abandoning individual deviation is
\[
\NW(\pi) = \sqrt[n]{(|C_1|-1)^{|C_1|} \cdot ... \cdot (|C_i|-1)^{|C_i|} \cdot ... \cdot (|C_k|-1)^{|C_k|}}\text.
\]
Suppose an agent is deviating from coalition \( C_i \) to coalition \(C_j \). Clearly, \( |C_i| \leq |C_j| \) before the deviation, because both coalitions are cliques and the deviation is only profitable for $i$ if the target coalition is larger.
The Nash welfare of the new partition $\pi'$ is
\begin{align*}
    &\NW(\pi') =\\
    &\sqrt[n]{(|C_1|-1)^{|C_1|} \cdot ... \cdot (|C_i|-2)^{|C_i|-1} \cdot ... \cdot (|C_j|)^{|C_j|+1} \cdot ... \cdot (|C_k|-1)^{|C_k|}}.
\end{align*}
We want to show \( NW(\pi') > NW(\pi) \), which is equivalent to showing that
\[
(|C_i|-1)^{|C_i|} \cdot (|C_j|-1)^{|C_j|} < (|C_i|-2)^{|C_i|-1} \cdot (|C_j|)^{|C_j|+1}\text.
\]
since all other factors remain unchanged.
Let \( p = |C_i| \) and \( q = |C_j| \), with \( p > 2 \) and \( p \leq q \). The inequality becomes
\[
(p-1)^p \cdot (q-1)^q < (p-2)^{p-1} \cdot q^{q+1}\text,
\]
which is equivalent to
\[
\frac{(p-1)^p}{(p-2)^{p-1}} < \frac{q^{q+1}}{(q-1)^q}.
\]
We define two functions $f(p)=\frac{(p-1)^p}{(p-2)^{p-1}}$ and $g(q)=\frac{q^{q+1}}{(q-1)^q}$. The goal is to show that $f(p)<g(q)$ for all $2 < p \leq q$. Observe that $g(q)=f(q+1)=f(p+k+1)$ for $k \geq 0$.
If we show that $f$ is monotonically increasing, then it follows that $f(p)< f(p+k+1)=g(q)$ and hence \( NW(\pi') > NW(\pi) \).
To prove that \( f \) is increasing for \( x > 2 \), consider its logarithm:
\[
h(x) = \log f(x) = x \log(x - 1) - (x - 1) \log(x - 2)
\]
Differentiating,

\[
h'(x) = \frac{d}{dx} \left[ x \log(x - 1) \right] - \frac{d}{dx} \left[ (x - 1) \log(x - 2) \right]=
\]

\[
 = \left(1 \cdot \log(x - 1) + x \cdot \frac{1}{x - 1} \right) - \left(1 \cdot \log(x - 2) + (x - 1) \cdot \frac{1}{x - 2} \right)=
\]

\[
 = \log(x - 1) + \frac{x}{x - 1} - \log(x - 2) - \frac{x - 1}{x - 2}=
\]

\[
 = \log\left(\frac{x - 1}{x - 2}\right) + \left( \frac{x}{x - 1} - \frac{x - 1}{x - 2} \right)=
\]

\[
 = \log\left(1 + \frac{1}{x - 2}\right) - \frac{1}{(x - 1)(x - 2)}\text.
\]
Since \( \log(1 + y) > \frac{y}{1 + y} \) for \( y > 0 \), and \( x > 2 \), we have
\[
\log\left(1 + \frac{1}{x - 2}\right) > \frac{1}{x - 1}, \quad \text{so}
\]
\[
h'(x) > \frac{1}{x - 1} - \frac{1}{(x - 1)(x - 2)} \ge 0\text.
\]
Hence, $h$ is increasing for $x>2$ and so is $f$. Therefore, $\NW(\pi')>\NW(\pi)$.
\end{proof}

\aegzero*
\begin{proof}
If the optimal Nash welfare is strictly positive, then there exists a partition into cliques without singletons. Every clique $K_m$ of the partition can be covered with $K_2$'s and $K_3$'s: if $m$ is even, take only $K_2$'s; if $m$ is odd, take one $K_3$ and partition the remaining players into $K_2$'s. So the graph admits a factor with $\{K_2,K_3\}$.
If the maximum size $\{K_2,K_3\}$-packing of $G_M$ is a $\{K_2,K_3\}$-factor, then every player gets utility of at least 1, meaning the Nash welfare is a product of strictly positive terms, i.e., it is strictly positive.
\end{proof}

\aegapprox*
\begin{proof}
Let $G_M$ be the mutual-friendship graph of the AEG $G$. We run the polynomial-time algorithm that computes a maximum-cardinality $\{K_2,K_3\}$-packing of $G_M$. If this packing is not a factor, we know by \Cref{lem:aeg_zero} that the optimal Nash welfare is $0$. In this case, any IR partition will do, e.g., the all-singletons-partition. Hence, assume now that the $\{K_2,K_3\}$-packing returned by the algorithm is a factor. Let $\pi$ be the partition induced by this factor, i.e., each copy of $K_2$ or $K_3$ forms a coalition in $\pi$.

The Nash welfare obtained by $\pi$ is at least $1$ because every agent has at least one friend in his coalition, rendering the geometric mean of agent utilities to have this lower bound. The Nash welfare obtained by the optimal partition $\pi^*$ is at most $n-1$, representing the case when $G_M$ is the complete graph and $\pi^*$ consists of the grand coalition, because then each agent has the maximum number of friends, i.e., $n-1$ in her coalition. We obtain the following.
\begin{equation*}
    \frac{\NW(\pi^*)}{\NW(\pi)}\le\frac{n-1}{1}=n-1
\end{equation*}
\end{proof}

\subsection*{Proofs for Section \ref{sec:afg}}

Here, we provide proofs related to AFGs.

\afgzero*
\begin{proof}
Suppose that there exists an agent $i^*$ who considers no other agent a friend.
In any coalition $S$ containing $i^*$ and some other agent, the utility of
$i^*$ is strictly negative, since every other member of $S$ contributes $-1$.
Therefore, in any IR partition $\pi$, agent $i^*$ must form a singleton,
and hence $u_{i^*}(\pi)=0$.
It follows that $\NW(\pi)=0$ for every IR partition $\pi$, and therefore
$\NW^*=0$.

Conversely, assume that every agent has at least one friend.
Consider the partition $\pi$ with only one coalition including all agents in it.
We have
\[
u_i(\pi) = n\cdot |F_i| - (n-1-|F_i|)= (n+1)\cdot |F_i| - (n-1) \;\ge\; (n+1)\cdot 1 - (n-1) = 2\]
for any $i\in N$, where $F_i$ is the set of friends of agent $i$.

We have $u_i(\pi)>0$ for all $i$, so $\pi$ is individually
rational and has strictly positive Nash welfare $\NW(\pi)>0$.
Therefore $\NW(\pi^*) \ge \NW(\pi)>0$.
\end{proof}

\afgapprox*
\begin{proof}
    First, we prove the claimed approximation ratio. The key is to show that each agent with a utility smaller than $\frac{n}{2}$ in the partition $\pi$ returned by the algorithm also has to have a utility smaller than $\frac{n}{2}$ in the optimal partition $\pi^*$.
    To this end, suppose $u_i({\pi})<\frac{n}{2}$ for some $i \in N$. We see that
    \begin{align*}
        &u_i(\pi)=n\cdot|F_i(\pi(i))|-|E_i(\pi(i))|<\frac{n}{2}
        \iff&|E_i(\pi(i))|>n\cdot\big(|F_i(\pi(i))|-\frac{1}{2}\big)\text,
    \end{align*}
    where $F_i(\pi(i))$ and $E_i(\pi(i))$ denote, respectively, the set of $i$'s friends and the set of $i$'s enemies within $i$'s coalition. 
The inequality cannot hold for $|F_i(\pi(i))|\geq 2$ as in such cases the number of enemies of agent $i$ in the coalition would need to be greater than $\frac{3}{2}n$.
    If $F_i(\pi(i))=0$, consider the cases $i\in\hat{U}$ and $i\notin\hat{U}$ which we show lead both to contradictions. First, assume $i\notin\hat{U}$. Then, $i$ had positive utility after phase~1. But it is impossible that this has changed because of the check in line~12. Second, assume $i\in\hat{U}$. Then, $i$ would have deviated in phase~2, and the same argument as in the first case holds.
    Hence, the only possible configuration is with $F_i(\pi(i))=1$, namely that $i$ has exactly one friend $f_i$ in $\pi(i)$ and at least $\frac{n}{2}+1$ enemies.

    Observe that there can only exist one coalition with such agents, otherwise we would exceed the total number of agents. We call this coalition $C$. So every agent with utility smaller than $\frac{n}{2}$ is in $C$. Suppose now there exists an agent $i$ in $C$ with $u_i(\pi)<\frac n2$ who has another friend in the game besides $f_i$. This other friend $o_i$ is in a different coalition $C'$ in $\pi$. Since $C$ contains more than half of all agents, it has to be a star on $G_M$ (see below), and $i\in U$. Hence, $i$ would have deviated to $C'$ in phase~2 of the algorithm before it terminates, a contradiction. Thus, each agent with utility less than $\frac n2$ in $\pi$ has only one friend in the entire game.

    Moreover, observe (as an invariant of \Cref{alg:matching-congestion}) that every coalition of $\pi$ is a connected component in $G_M$. More specifically, it is either a triangle or a star. The first case of triangles are precisely the triangles present after phase~1. The second case of stars comes from observing that all agents in $U\cap\hat{C}$ are pairwise enemies for any coalition $\hat{C}\in\pi$, and each coalition of $\pi$ which contains at least one agent in $U$ contains exactly one agent not in $U$.

    To show the claim that all agents in $U\cap\hat{C}$ are indeed pairwise enemies, first observe that this clearly holds for all initially unassigned agents ($\hat{U}$) as otherwise we directly have a contradiction to the maximality of the initial packing. Further, in line 9, only agents get added who are enemies with all agents in $\hat{U}$. Hence, the only possible option would be that two agents $i$ and $j$ in $U\setminus\hat{U}$ are friends. Assume now that such a pair of agents exists. Call the respective partners of $i$ and $j$ in the initial packing $p_i$ and $p_j$. Both have at least one friend in $\hat{U}$ by the check in line~8. If those friends are distinct, this again contradicts the maximality of the initial packing as there would be an augmenting path with them. It remains to consider the case that $p_i$ and $p_j$ both only have one friend in $\hat{U}$ and it is the same agent. Let that agent be called $p_{ij}$. Assume further, w.l.o.g., that at some point, $i$ joined $j$ in $\hat{C}$. There must have been a reason for this deviation. The possibilities are: First, some enemy of $i$ (who has to be a friend of $p_i$ then) joined the pair $i$, $p_i$. Or, second, some friend of $i$ joined the pair $j$, $p_j$ (and hence has to be a friend of $j$ or $p_j$). We can treat both cases the same way because the important point is merely that some agent, let us denote her by $l_1$ is friends with a member of one of the pairs. Moreover, $l_1$ cannot be in $\hat{U}$ because we ensured $p_{ij}$ is the only such friend of $p_i$ and $p_j$ while $j$ has no such friend at all. The high-level argument is now that the constellation $i$, $p_i$, $p_{ij}$, $p_j$, $j$ forms a blossom, and we can find an augmenting path along $l_1$. This goes as follows. $l_1\notin\hat{U}$ means $l_1$ is part of the initial packing. If $l_1$ is part of a triangle, we directly have an augmenting path. So $l_1$ has just one partner $k_1$. Since $l_1$ eventually left $k_1$ to join $j$ and $p_j$, there had to be another agent $l_2$ who joined $k_1$ beforehand. Otherwise the deviation by $l_1$ would have abandoned $k_1$. Now the arguments start to repeat. If $l_2\in\hat{U}$ or $l_2$ is part of a triangle in the initial packing, we have an augmenting path. So $l_2$ has just one partner in the initial packing, let us call him $k_2$. By repeating the arguments above, we get pairs of agents $l_{x+1}$ and $k_{x+1}$ from $l_x$ and $k_x$ until eventually some $l_x$ is in $\hat{U}$ or in a triangle of the initial packing. This has to happen since there are only finitely many agents. Finally, we have found an augmenting path, leading to a contradiction.

    Recall now the coalition $C$ introduced above. Since $C$ has to be a star coalition as just proven (it clearly cannot be a triangle), let us call the common friend of all members $f$. All agents in $C$ receive the same utility and hence by the above reasoning have $f$ as their only friend in the game. Thus, in the optimal partition $\pi^*$, all members of $C$ still have to be in a coalition with $f$ and hence get at most the utility they get in $\pi$.

\begin{equation*}
        u_i({\pi})<\frac{n}{2} \implies u_i({\pi^*})<\frac{n}{2}
    \end{equation*}

    We now analyze the approximation ratio: 
    \[
    \begin{split}
    \frac{\NW(\pi^*)}{\NW(\pi)}
    &=\frac{\sqrt[n]{\prod_{i\in N}u_i(\pi^*)}}{\sqrt[n]{\prod_{i\in N}u_i(\pi)}}
    =\sqrt[n]{\prod_{i\in N}\frac{u_i(\pi^*)}{u_i(\pi)}} \\
    &= \sqrt[n]{\prod_{i:\,u_i(\pi)\ge \frac{n}{2}}\frac{u_i(\pi^*)}{u_i(\pi)}}
    \cdot
    \sqrt[n]{\prod_{i:\,u_i(\pi)< \frac{n}{2}}\frac{u_i(\pi^*)}{u_i(\pi)}}\\
    &< \sqrt[n]{\prod_{i:\,u_i(\pi)\ge \frac{n}{2}}\frac{n^2}{\left(\frac{n}{2}\right)}}
    \cdot
    \sqrt[n]{\prod_{i:\,u_i(\pi)< \frac{n}{2}}\frac{\left(\frac{n}{2}\right)}{2}}
    \end{split}
    \]
    since $u_i(\pi^*)\leq n(n-1)<n^2$ in the full graph, $u_i(\pi)\geq2$ (every agent is in a coalition with at least one friend and at most $n-2$ enemies) and we proved that $u_i(\pi)<\frac{n}{2} \implies u_i(\pi^*)<\frac{n}{2}$. Let now $k$ denote the number of agents $i$ such that $u_i(\pi)\geq\frac{n}{2}$. Then,
    \begin{equation*}
        \frac{\NW(\pi^*)}{\NW(\pi)}<\sqrt[n]{(2n)^k\cdot\left(\frac{n}{4}\right)^{n-k}}=\sqrt[n]{8^k\cdot\left(\frac{n}{4}\right)^n}=\frac{n}{4}\sqrt[n]{8^k}<2n
    \end{equation*}
    where the last inequality holds since $k<n$.

    We proceed to prove the claimed polynomial running time. To this end, first note that phase~1 of the algorithm clearly takes polynomial time only, including finding a maximal $\{K_2,K_3\}$-packing \citep{HeKi84a}. For phase~2, we use utilitarian welfare as a potential function. Every (non-abandoning) Nash deviation increases the utilitarian welfare since the game is symmetric \citep{BoJa02a}. Furthermore, this increase has to be at least by $1$ because all valuations (and thus all utilities) are integers. Before phase~2, utilitarian welfare is non-negative, and the possible utilitarian welfare in the game is upper-bounded by $n^3$. As a result, there can occur at most $n^3$ deviations in phase~2, and the entire algorithm runs in polynomial time.

\end{proof}

\subsection*{Proofs for Section \ref{sec:hardness}}

Here, we provide the proof for hardness of approximation.

\approxhardness*
\begin{proof}
We reduce from \textsc{Uncapacitated Two-Sided Nash Welfare} (\textsc{UTSNW}), which is NP-hard to approximate
within a factor of $1.0000759$ even under additive valuations in $\{0,1,2\}$ \citep{GaMu23a,GSVY25a}.
Given an instance with workers $W$, firms $F$, and nonnegative (possibly asymmetric) utilities
$U_{\mathit{wf}},U_{\mathit{fw}}\ge 0$, we construct an ASHG on the same agent set $N:=W\cup F$ with additive valuations
\[
v_{\mathit{wf}}:=U_{\mathit{wf}},\qquad v_{fw}:=U_{\mathit{fw}}\qquad (w\in W,f\in F),
\]
\[
v_{\mathit{ww'}}:=0\qquad (w\neq w'\in W),
\qquad
v_{ff'}:=-H\qquad (f\neq f'\in F),
\]
where
\[
H > \max_{f\in F}\;\sum_{x\in W} v_{fx}
\]
For a partition $\pi$ of $N$, each agent's utility is
$u_i(\pi)=\sum_{j\in \pi(i)\setminus\{i\}} v_{\mathit{ij}}$, and we maximize the geometric mean
$\NW(\pi)=\sqrt[n]{\prod_{i\in N}u_i(\pi)}$ over IR partitions.

First, any IR partition contains at most one firm per coalition: if $f\neq f'$ are in the same coalition, then
\[
u_f(\pi)\le -H + \sum_{x\in W} v_{\mathit{fx}} < 0\text,
\]
contradicting IR.
Now consider any \textsc{UTSNW} allocation $\mu:W\to F$. Define the partition $\pi_\mu$ by creating, for each firm
$f$, the coalition $\{f\}\cup\{w\in W:\mu(w)=f\}$. Since $v_{\mathit{ww'}}=0$, we have for all $w\in W$ and $f\in F$:
$u_w(\pi_\mu)=v_{w,\mu(w)}=U_{w,\mu(w)}
\qquad\text{and}\qquad
u_f(\pi_\mu)=\sum_{w:\mu(w)=f} v_{\mathit{fw}}=\sum_{w:\mu(w)=f} U_{\mathit{fw}}$,
hence $\NW(\pi_\mu)=\NW(\mu)$.

Conversely, let $\pi$ be any IR partition. By the previous paragraph, each coalition contains at most one firm.
Define $\mu_\pi:W\to F$ by assigning each worker $w$ to the (unique) firm $f$ in her coalition. Then, again by
construction,
\[
u_w(\pi)=U_{w,\mu_\pi(w)}
\qquad\text{and}\qquad
u_f(\pi)=\sum_{w:\mu_\pi(w)=f} U_{fw},
\]
so $\NW(\mu_\pi)=\NW(\pi)$.
Therefore the reduction preserves the objective value exactly and is computable in polynomial time, so any polynomial-time $\alpha$-approximation for ASHG Nash welfare would yield a polynomial-time $\alpha$-approximation
for \textsc{UTSNW}. Taking $\alpha<1.0000759$ contradicts \citep[Theorem~6.1]{GSVY25a} unless
$\P=\NP$. Hence it is \NP-hard to approximate Nash welfare for ASHGs within a factor smaller than
$1.0000759$.
\end{proof}

\subsection*{Proofs for Section \ref{sec:restrict}}

Here, we provide proofs related to AEGs with bounded number or size of coalitions.

\twocoalitions*
\begin{proof}
Let $G_M$ be the mutual-friendship graph of the AEG on $n$ agents and let $H=\overline{G_M}$ be its complement.
Since we are allowed to use at most two coalitions, an IR partition exists if and only if $G_M$ is the union of at most
two cliques, i.e., $H$ is bipartite. A bipartition can be found in polynomial time by BFS or DFS.
Fix any bipartition $(X,Y)$ of $H$ and write $s:=\min\{|X|,|Y|\}\le \frac n2$. In the corresponding partition of agents into two coalitions
of sizes $s$ and $n-s$, every agent in a coalition of size $t$ has utility $t-1$, hence the Nash welfare of a partition is determined by $s$. We would therefore just write $\NW(s)$. Further, the Nash welfare strictly increases with the function $F(s):=\sqrt[n]{\NW(s)}$. We have
\[
F(s)=(s-1)^s\cdot(n-s-1)^{n-s}.
\]
If $s=0$ then the grand coalition is trivially optimal. If $s=1$ then $F(s)=0$ since some agent has utility $0$. Thus the nontrivial case is $2\le s\le n-2$.
Now observe that different connected components of a bipartite graph can be flipped independently.
Let the connected components of $H$ be $C_1,\dots,C_c$. For each $C_i$, compute a bipartition
$(A_i,B_i)$ by BFS/DFS and assume w.l.o.g.\ $|A_i|\le |B_i|$. Any global bipartition of $H$ is obtained by choosing,
for each component, whether $(A_i,B_i)$ is kept or swapped, so the set of possible values of $s=\min\{|X|,|Y|\}$
depends only on these flips. We claim that $F(s)$ is strictly decreasing for $s\in[2,\frac n2]$ increasing, hence the Nash welfare
is maximized by minimizing $s$.

To see this, let $g(s):=\ln(F(s))= s\ln(s-1) + (n-s)\ln(n-s-1)$ for $2\le s\le \frac n2$. Then,
\[
g'(s)= \ln(s-1)+\frac{s}{s-1} - \ln(n-s-1) - \frac{n-s}{n-s-1},
\]
and
\[
g''(s)= \frac{s-2}{(s-1)^2} + \frac{n-s-2}{(n-s-1)^2}.
\]
For $2\le s\le \frac n2$ we have $g''(s)>0$, so $g'$ is strictly increasing. Moreover $g'(\frac n2)=0$ by symmetry, and so $g'(s)<0$ for all $s<\frac n2$, implying that $g(s)$, therefore $F(s)$, and therefore $\NW(s)$ are strictly decreasing in this region.

Therefore, among all bipartitions, an optimal one is obtained by making the smaller side as small as possible, very much like the consequences of \Cref{prop:non-abandon}.
This is achieved by taking
\[
X := \bigcup_{i=1}^c A_i \qquad\text{and}\qquad Y := \bigcup_{i=1}^c B_i,
\]
when for each component the smaller side $A_i$ is assigned to the same coalition. The resulting partition
into at most two coalitions maximizes Nash welfare (in the degenerate cases of $s=1$ it yields the maximum of $\NW(\pi^*)=0$).
All steps (bipartite test, component decomposition, BFS colorings, and computing sizes) can be performed in polynomial time
(in fact $\mathcal O(n+m)$, where $m$ is the number of edges in the graph), so the maximal Nash welfare for AEGs with at most two allowed coalitions can be efficiently computed.
\end{proof}

\hardnessnumcoal*
\begin{proof}
We reduce from $k$-\textsc{Coloring}, NP-complete problem for every fixed $k\ge 3$, as shown by \citet{GaJo79a}. Given an instance $H$, let $G=\overline{H}$ and obtain $G'$ by adding $k$ new vertices $u_1,\dots,u_k$ such that:
(i) each $u_j$ is adjacent to every vertex of $G$; and
(ii) the $u_j$ are pairwise non-adjacent.
The AEG instance is given by $G'$, considering it a friendship graph.

\smallskip
\noindent(\(\Rightarrow\)) If $H$ is $k$-colorable with color classes $C_1,\dots,C_k$, then each $C_j$ is a clique in $G$. Put $u_j$ together with $C_j$. This yields exactly $k$ cliques in $G'$, each of size at least $2$, so $\NW>0$.

\smallskip
\noindent(\(\Leftarrow\)) Suppose there is a partition $\pi$ of $V(G')$ into at most $k$ cliques with $\NW>0$. No singleton is allowed. Since the $u_j$ are pairwise non-adjacent, no clique contains two distinct $u_j$. Each $u_j$ must share a clique with at least one original vertex of $G$ to avoid being a singleton. Hence $\pi$ has at least $k$ cliques and, under the constraint “at most $k$,” it has exactly $k$ cliques, each containing exactly one $u_j$ plus a subset of $V(G)$ that is a clique in $G$. Removing the $u_j$ yields a partition of $V(G)$ into $k$ cliques, i.e., a partition of $V(H)$ into $k$ independent sets, so $H$ is $k$-colorable.

The reduction is polynomial and preserves YES/NO answers. Therefore the decision problem is \NP-complete for every fixed $k\ge 3$.
\end{proof}

\algosizecoal*
\begin{proof}
    Feasible partitions are perfect matchings of the mutual-friendship graph $G_M$ representing the game.
    Thus, if $n$ is odd, we immediately know the optimal Nash welfare is $0$.
    Assume $n$ is even in the following.
    \begin{itemize}
        \item In the binary case (friendship edges of weight $1$), any perfect matching yields $\NW=1$. Deciding existence and finding one reduces to \emph{maximum cardinality matching}.
        \item With asymmetric valuations, restricting to $s=2$ yields an objective equal to the product over edges $\{i,j\}$ in the matching of weight $v_{ij}v_{ji}$. Maximizing the product is equivalent to maximizing the sum of $\log(v_{ij}v_{ji})$, hence to a \emph{maximum weight matching} with edge weights $w_{ij}=\log(v_{ij}v_{ji})$.

    \end{itemize}
    Both maximum cardinality and maximum weight matching are solvable in polynomial time (e.g., Edmonds’ blossom algorithm \citep{Edmo65a} and its weighted extensions \citep{GaTa91a}).
\end{proof}

\hardnesscoalsize*
\begin{proof}
We first show hardness for $s=3$ and then for $s>3$ with an analogous proof.
Reduce from \textsc{PartitionIntoTriangles}: given a graph $G$, decide whether $V(G)$ can be partitioned into vertex-disjoint triangles.

Let $\mathcal{A}$ be as in the statement. Since $\mathcal{A}$ is finite, $|\mathcal{A}|\ge 2$, and $\max \mathcal{A}>0$, we can choose:
a positive valuation $\alpha \in \mathcal{A}$ such that $\alpha = \max(\mathcal{A})$ and a valuation $\beta < \alpha$.
Build a symmetric ASHG on the agent set $N := V(G)$ by setting, for every distinct $i,j \in N$,
\[
v_{ij} = v_{ji} :=
\begin{cases}
\alpha & \text{if } \{i,j\} \in E(G),\\
\beta  & \text{otherwise.}
\end{cases}
\]
Thus the underlying mutual-friendship graph coincides with $G$, and edges of $G$ correspond to symmetric pairs of valuation~$\alpha$.

If $G$ admits a triangle factor, then there exists a partition of $V$ into $K_3$’s, every agent has utility $2\alpha$, and $\NW(\pi^*)=2\alpha$. Conversely, if no triangle factor exists, then in every feasible partition some player in a coalition has at most 1 friend, getting utility at most $\alpha$, in particular, $\NW(\pi^*)\le 2\alpha$, with equality only when all coalitions are triangles. Therefore deciding whether $\NW(\pi^*)=2\alpha$ solves \textsc{PartitionIntoTriangles}. Since \textsc{PartitionIntoTriangles} is NP-complete, maximizing $\NW$ is NP-hard for $s=3$.

For any fixed $s>3$, reduce from \textsc{PartitionInto$K_s$} (NP-complete problem as shown by \citet{HeKi84a}) 
using the same construction. If $G$ has a $K_s$-factor then $\NW(\pi^*)=(s-1)\alpha$; otherwise the optimum is strictly smaller. Hence the optimization problem is NP-hard for all $s\ge 3$.
\end{proof}

\aegapproxcoalsize*
\begin{proof}
If the packing fails to cover all vertices, $\NW(\pi^*)=0$, so the approximation ratio is $1$.
Otherwise we output a $\{K_2,K_3\}$-factor, let us call $\pi$ the corresponding partition. Every agent in a $K_2$ coalition has utility of $1$; every agent in a $K_3$ coalition has utility of $2$, hence $\NW(\pi)\geq 1$. On the other hand, under the size bound $s$ no agent can exceed utility $s-1$, so $\NW(\pi^*)\le s-1$. Therefore
\[
\frac{\NW(\pi^*)}{\NW(\pi)}\ \le\ \frac{s-1}{1}\ =\ s-1.
\]
For $s=3$ the bound is tight: there are graphs admitting both a perfect matching and a triangle factor; a $\{K_2,K_3\}$-factor returned by the packing routine may be a perfect matching, giving $\NW(\pi)=1$ while the optimal triangle factor gives $\NW(\pi^*)=2$.
\end{proof}

\subsection*{Proofs for Section \ref{sec:stability}}

Here, we provide the proof for implied CNS in symmetric ASHGs.

\nashcns*
\begin{proof}
    Note that if $\NW(\pi^*)>0$, then any partition $\pi$ that achieves a finite approximation factor with respect to $\pi^*$ must also satisfy $\NW(\pi)>0$. Equivalently, we must have $u_i(\pi)>0$ for all $i\in N$, since otherwise $\NW(\pi)=0$ and the approximation ratio would be unbounded.

    Suppose by contradiction that $\pi$ is not CNS. Then, there exist an agent $i$ and a different partition $\pi'$ that can be reached via the contractual Nash deviation of agent $i$. Hence, it holds that $u_j(\pi')\geq u_j(\pi)$ for all $j\in \pi(i)$. Fix any $j\in \pi(i)\setminus \{i\}=\pi'(j)$. Since only agent $i$ is deviating, this is the only change in $j$'s coalition. Thus, 
    $u_j(\pi(i)) = \sum_{k \in \pi(i) \setminus \{j\}} v_{jk}
    \quad\text{and}\quad
    u_j(\pi'(j)) = \sum_{k \in \pi'(j) \setminus \{j\}} v_{jk}
    = u_j(\pi(i)) - v_{ji}.$
    
    The condition $u_j(\pi')\geq u_j(\pi)$ then becomes
\[
u_j(\pi(i)) - v_{ji} \;\ge\; u_j(\pi(i))
\quad\Longrightarrow\quad
-v_{ji} \;\ge\; 0
\quad\Longrightarrow\quad
v_{ji} \;\le\; 0.
\]

By symmetry, $v_{ji} = v_{ij}$, so we obtain $v_{ij} \le 0$ for every $j \in \pi(i) \setminus \{i\}$.
Hence
\[
u_i(\pi(i)) \;=\; \sum_{j \in \pi(i) \setminus \{i\}} v_{ij} \;\le\; 0.
\]
On the other hand $\pi$ is IR, so $u_i(\pi(i)) \ge 0$.
Combining these inequalities yields
\[
u_i(\pi(i)) = 0.
\]
This contradicts the assumption that $u_i(\pi(i)) > 0$ for all $i \in N$.
Therefore no such deviation from $\pi$ can exist, and $\pi$ must be CNS.
\end{proof}

\end{document}